\documentclass[runningheads]{llncs}

\usepackage{graphicx}
\usepackage{siunitx}
\usepackage{comment}
\usepackage{amsmath}
\usepackage{amsfonts}
\usepackage{amssymb}
\usepackage{hyperref}
\usepackage{algorithm}
\usepackage{sidecap, placeins}
\usepackage[margin=3pt, size=small]{subcaption}
\usepackage[noend]{algpseudocode}
\usepackage{color}
\usepackage{enumitem}

\graphicspath{ {./images/} }

\setlist[itemize]{topsep=\parskip}

\let\emptyset\varnothing

\DeclareMathOperator*{\argmax}{arg\,max}
\DeclareMathOperator*{\argmin}{arg\,min}

\begin{document}
\title{COSTA: Communication-Optimal Shuffle and Transpose Algorithm with Process Relabeling} \titlerunning{COSTA: Comm-Optimal Shuffle and Transpose Alg. with Process Relabeling}


%
%
\author{Marko Kabi\'{c}\inst{1,2} \and
Simon Pintarelli \inst{1,2} \and
Anton Kozhevnikov \inst{1,2} \and
Joost VandeVondele \inst{1,2}}
\authorrunning{M. Kabi\'{c}, S. Pintarelli, A. Kozhevnikov, J. VandeVondele}
%
%
\institute{ETH Zürich, Switzerland \and
Swiss National Supercomputing Centre (CSCS), Switzerland}
\maketitle              
\begin{abstract}
Communication-avoiding algorithms for Linear Algebra have become increasingly popular, in particular for distributed memory architectures. In practice, these algorithms assume that the data is already distributed in a specific way, thus making data reshuffling a key to use them. For performance reasons, a straightforward all-to-all exchange must be avoided.

Here, we show that process relabeling (i.e. permuting processes in the final layout) can be used to obtain communication optimality for data reshuffling, and that it can be efficiently found by solving a Linear Assignment Problem (Maximum Weight Bipartite Perfect Matching). Based on this, we have developed a Communication-Optimal Shuffle and Transpose Algorithm (COSTA): this highly-optimised algorithm implements $A=\alpha\cdot \operatorname{op}(B) + \beta \cdot A,\  \operatorname{op} \in \{\operatorname{transpose}, \operatorname{conjugate-transpose}, \operatorname{identity}\}$ on distributed systems, where $A, B$ are matrices with potentially different (distributed) layouts and $\alpha, \beta$ are scalars. COSTA can take advantage of the communication-optimal process relabeling  even for heterogeneous network topologies, where latency and bandwidth differ among nodes. Moreover, our algorithm can be easily generalized to even more generic problems, making it suitable for distributed Machine Learning applications. 
The implementation not only outperforms the best available ScaLAPACK redistribute and transpose routines multiple times, but is also able to deal with more general matrix layouts, in particular it is not limited to block-cyclic layouts. Finally, we use COSTA to integrate a communication-optimal matrix multiplication algorithm into the CP2K quantum chemistry simulation package. This way, we show that COSTA can be used to unlock the full potential of recent Linear Algebra algorithms in applications by facilitating interoperability between algorithms with a wide range of data layouts, in addition to bringing significant redistribution speedups.

\end{abstract}
\section{Introduction}

Communication-avoiding algorithms for Linear Algebra have become increasingly popular recently, in particular for distributed memory architectures. In practice, these algorithms usually assume the data is already distributed in a specific way. For example, COSMA \cite{cosma}, a communication-optimal matrix multiplication algorithm, natively uses a specialised blocked data layout which depends on matrix shapes and the available resources. Similarly, CARMA \cite{carma} is a recursive communication-avoiding algorithm which requires a block-recursive data layout also depending on matrix shapes and the available resources. On the other hand, most of the software packages for scientific applications like CP2K \cite{cp2k} use ScaLAPACK API \cite{scalapack} which assumes a block-cyclic matrix layout. Hence, the data redistribution (i.e.\ reshuffling) becomes necessary in order to integrate and use these efficient algorithms within the well-established software packages. Performing an all-to-all communication would violate the communication-optimality of these algorithms and must be avoided for performance reasons.

Another example where the data reshuffling is needed is to achieve the optimal performance of existing ScaLAPACK routines. It is known that the performance of ScaLAPACK highly depends on the block size which determines the matrix layout. The optimal block size depends on the target machine and the computational kernel being used \cite{blocks}. Therefore, the data reshuffling might be needed to achieve the optimal block size on a specific system. ScaLAPACK provides the routine pxgemr2d for data reshuffling \cite{pxgemr2d}, but it is limited to block-cyclic layouts.

Here we present COSTA: an algorithm that resolves all these problems: 1) it minimizes the communication-cost of data reshuffling by process relabeling in the target layout 2) it can handle arbitrary grid-like matrix layouts which are not necessarily block-cyclic. Moreover, both row-major and col-major ordering of blocks is supported; 3) it can also transform the data while reshuffling (e.g. transpose or multiply by a scalar). 4) efficiently utilizes the overlap of communication and computation (i.e. the transformation); 5) provides the batched version, which can transform multiple layouts at once, while significantly reducing the latency. COSTA stands for Communication-Optimal (Re-)Shuffle and Transpose Algorithm and refers to the matrix operation that the algorithm implements: $A\leftarrow \alpha \cdot \operatorname{op}(B) + \beta \cdot A$, where $\operatorname{op} \in \{\operatorname{transpose}, \operatorname{conjugate-transpose}, \operatorname{identity}\}$, $A, B$ are distributed matrices and $\alpha, \beta$ scalars. 

The idea of relabeling processes in order to reduce the communication-cost has already been studied in~\cite{redistribution}. However, their model has the following limitations: 1) it implicitly assumes all data pieces (i.e.\ items) are of the same size; 2) each data piece is assumed to belong to a single process; 3) data transformation (e.g. transpose, or multiplication by a scalar) during reshuffling is not considered; 4) the model does not take into account the data locality, e.g. how the local data is stored in the memory; 5) the latency and the bandwidth of all processes is assumed to be the same.

\begin{sloppypar}
In this paper, our contribution is twofold. First, we develop a generalization of the model presented in \cite{redistribution} for finding the Communication-Optimal Process Relabeling that resolves the above-mentioned limitations and is not limited to Linear-Algebra applications. Then, using this model we develop the COSTA algorithm that resolves the limitations of ScaLAPACK \texttt{pxgemr2d} and \texttt{pxtran} routines and outperforms them even when no process relabeling is used. 
\end{sloppypar}

\section{Preliminaries and Notation}

For an arbitrary set $s=\{b_0, b_1, \dots, b_{n-1}\}$, $\lvert s\rvert=n$ denotes its size. For a Cartesian product between two sets $s_1$ and $s_2$ we write $s_1\times s_2$. We define the range as $[n] :=\{0, 1, \dots, n-1\}$. A square matrix $X$ of dimension $n$ is denoted by $X = [[x_{ij}]]_{i, j \in [n]}$ or equivalently $X = [[x_{ij}]]_{0\leq i, j < n}$. We might treat a matrix as a set of its entries, in particular, we might write $x_{ij} \in X$ for $i,j \in [n]$.

A graph $G=(V, E)$ has vertices $V$ and edges $E\subseteq (V\times V)$. A weighted graph is a graph in which each edge is assigned some weight. A bipartite graph with partitions $U$ and $V$ is a graph $G=(U\cup V, E)$ where $E\subseteq U \times V$. If $V=\{v_0, v_1, \dots v_{n-1}\}$ and $V', V''$ are two identical copies of $V$, e.g. $V'=\{v'_0, v'_1, \dots, v'_{n-1}\}$ and $V''=\{v''_0, v''_1, \dots, v''_{n-1}\}$, then we abuse the notation and write $G=(V, V, E)$ to denote the bipartite graph $G=(V', V'', E)$. Moreover, for an edge in this graph, we write $(v_i, v_j)$ instead of $(v'_i, v''_j)$. In a bipartite graph $G=(U, V, E)$, we define the set of all left neighbors of some vertex $v \in V$ as $\mathcal{N}^L_G(v) := \{u\in U : (u, v) \in E\}$ and similarly the set of all right neighbors of some vertex $u \in U$ as $\mathcal{N}^R_G(u) := \{v\in V: (u, v) \in E\}$. Two edges are \emph{adjacent} if they have a common vertex. A matching $M\subseteq E$ is a subset of non-adjacent edges. The weight of the matching is the sum of weights of its edges. A perfect matching of $G$ is a matching that covers every vertex of $G$.

\textbf{Machine Model.}
We assume a distributed memory setting with multiple processes, each having its own private memory. Our model best corresponds to the MPI parallel computing model, where our term \emph{process} corresponds to an MPI rank. We use the term \emph{process} instead of \emph{processor} to avoid the confusion arising when MPI is run on many-core architectures with different process affinity bindings. \emph{Local data} of some process, is the data residing in its private memory, whereas \emph{global data} is the union of local data across all processes. All data, that is not local to some rank is called \emph{remote}. We say that each process \emph{owns} its local data. 

\textbf{Data Package, Block and Volume.}
Let $p_i$ and $p_j$ be two arbitrary processes and let $s = \left\{b_0, b_1, \dots, b_{\lvert s \rvert}\right\}$ be a set of all data pieces that should be sent from $p_i$ to $p_j$. Each data piece $b$ might contain the information about its memory layout (e.g. if it is stored as a 2D block), data locality (e.g. stride, padding, alignment), memory ordering (e.g. row- or column-major) and similar. Each $b \in s$ is called a \emph{block} and $s$ a \emph{package}. A \emph{block volume} $V(b)$ is the size of block $b$ in bytes and similarly, a \emph{package volume} $V(s)$ is the sum of all block volumes it contains: $V(s) = \sum_{b\in s} V(b)$. 

\section{Communication Cost Function}
\label{sec:comm_cost}
\begin{sloppypar}
The communication cost $w(p_i, p_j, s)$ represents the cost of sending the package $s=\{b_0, b_1,\dots,b_{\lvert s \rvert}\}$ from process $p_i$ to process $p_j$. Formally, if $P=\left\{p_0, p_1, \dots, p_{n-1}\right\}$ is the set of all processes and $S$ the set of all packages that are to be exchanged, then the communication cost function is defined as a function $w: P \times P \times S \mapsto \mathbb{R}$. Specifically, $w(p_i, p_j, \emptyset)=0$ for any $p_i, p_j \in P$.
\end{sloppypar}
For an arbitrary $p_i, p_j \in P$ and $s \in S$, in its simplest form, $w$ can be defined as follows:
\begin{equation}
w(p_i, p_j, s) = \begin{cases}
V(s), &p_i \neq p_j \\
0, &\text{otherwise}
\end{cases}
\label{eq:locally_free_w}
\end{equation}
This cost function considers all local communication free, whereas a remote communication cost is equal to the data volume. We will refer to this communication cost as \emph{locally-free-volume-based} cost.

\noindent Alternatively, the cost function can also include the following factors: 
\begin{itemize}
    \item \textbf{Network Topology:} $w$ can take into account the physical topology e.g. using some of the bandwidth-latency models \cite{latency_bandwidth}. If $L(p_i, p_j)$ is the latency cost between $p_i$ and $p_j$, and $B(p_i, p_j)$ be the bandwidth cost per data unit, $w$ could be defined as $w(p_i, p_j, s) = L(p_i, p_j) + B(p_i, p_j) \cdot V(s).$
    \item \textbf{Transformation cost:} if the data also needs to be transformed on-the-fly, while being sent, e.g. the data should be transposed or multiplied by a scalar, then the cost of this transformation can also be included. For example, $w$ can be defined as $w(p_i, p_j, s) = c \cdot \sum_{b \in s}I_{T}(b) \cdot \lvert b \rvert,$ where $I_{T}(\cdot)$ is an indicator function if the piece of data $b$ should be transformed (e.g. multiplied by a scalar or transposed) while being sent from $p_i$ to $p_j$ and $c$ is a constant that determines the complexity of the transformation to be applied to $b$.
    \item \textbf{Data Locality:} the way how each piece of data $b \in s$ is stored in the memory can also be taken into account. For example, if $b$ is a 2D block, then $w$ can take into account its stride, padding, alignment and similar. 
\end{itemize}

\subsection{Communication Graph}
\label{sec:comm_cost_graph}

A communication graph describes the communication pattern, i.e.,\ which data the processes are going to exchange.

Let $P=\left\{p_0, p_1, \dots, p_{n-1}\right\}$ be the set of all processes and $S = [[S_{ij}]]_{0\leq i, j < n}$ be the set of all packages that are to be exchanged, where $S_{ij}$ corresponds to the package to be sent from $p_i$ to $p_j$. 

We can represent this communication pattern with an undirected, bipartite graph $(P, P, E)$ with two identical partitions $P$ and the set of edges $E$, defined as follows:
$$E = \left\{(p_i, p_j) : (p_i, p_j) \in (P\times P)\land S_{ij} \neq \emptyset \right\}$$
\noindent We now formally define a \emph{communication graph} as an ordered tuple:
\begin{equation}
G = (P, E, S)
\label{def:G}
\end{equation}
\noindent If $w: P \times P\times S$ is a communication-cost function, then under $w$, each edge $(p_i, p_j) \in E$  has weight $w(p_i, p_j, S_{ij})$. In the same manner, the \emph{total communication cost} of graph $G$, denoted by $W(G)$, is defined as the sum of the weights of all edges:
\begin{equation}
    W(G) = \sum_{(p_i, p_j) \in E} w(p_i, p_j, S_{ij}).
    \label{def:WG}
\end{equation}

\section{Communication-Optimal Process Relabeling (COPR)}

In this section, we first formally define the COPR and then show that it can be formulated as a well-known Linear Assignment Problem (LAP). Then, we discuss the current state-of-the-art algorithms to solve the LAP that can also be used to find the COPR.

\subsection{The Formal Definition}

Let $G=(P, E, S)$ be a communication graph (see Equation~\eqref{def:G}) on processes $P$, with edges $E$ and data set $S = [[S_{ij}]]_{0\leq i,j < n}$.

In order to reduce the total communication cost $W(G)$, as defined in Equation~\eqref{def:WG}, we want to relabel the processes. Relabeling $p_j$ to $p_i$ makes their communication become local, hence potentially reducing the communication cost. We will first formally define these terms and then aim to find the process relabeling that minimizes the total communication cost. 

\begin{definition} (Process Relabeling) Let $P=\left\{p_0, p_1, \dots, p_{n-1}\right\}$ be a set of processes. A process relabeling $\sigma$ is a permutation of indices $[n]$, implicitly mapping each $p_i$ to $p_{\sigma(i)}$.
\end{definition}

Applying a process relabeling $\sigma$ to graph $G$ under communication-cost function $w$ yields the relabeled communication graph $G_{\sigma}$ which we define as follows.

\begin{definition} \label{def:RG}
(Relabeled Graph) Let $G=(P, E, S)$ be a communication graph and $\sigma$ be a process relabeling of $P$. The relabeled graph $G_{\sigma}$ is a communication graph $G_{\sigma}=(P, E', S')$, where:
\begin{align*}
&E' = \left\{(p_{i}, p_{\sigma(j)}) : (p_i, p_j) \in E\right\} \\
&S'_{i, \sigma(j)} = S_{ij}  \text{ , for each $(p_i, p_j) \in E$.}
\end{align*}
\end{definition}
\begin{remark}
Observe that the initial graph $G$ is isomorphic to $G_{id}$ where $id(\cdot)$ is the identity permutation $id(i)=i$ for all $i\in[n]$.
\end{remark}
We say that graph $G_\sigma$ is \emph{induced} by $\sigma$ relabeling. Note that after relabelling $j\rightarrow \sigma(j)$, the processes $p_i$ still has to send $S_{ij}$ to $p_{\sigma(j)}$, as before relabelling. If the communication between $p_i$ and $p_{\sigma(j)}$ is faster than between $i$ and $j$, this relabeling might reduce the communication cost.

Finally, we define the communication-optimal process relabeling (COPR) as the relabeling which yields the graph $G_\sigma$ with minimal cost.

\begin{definition}
(COPR) Let $G=(P, E, S)$ be a communication graph and $w: P \times P \times S\mapsto \mathbb{R}$ be a communication cost function. A communication-optimal process relabeling $\sigma_{opt}$ w.r.t. the cost function $w$ is defined as:
$$\sigma_{opt} = \argmin_{\sigma} W(G_{\sigma})$$
\label{def:COPR}
\end{definition}

\subsection{COPR as Linear Assignment Problem}

In this section we show how finding the communication-optimal process relabeling (COPR) from Definition~\ref{def:COPR} can be reduced to solving the Linear Assignment Problem (LAP) \cite{assignment_revised}. The LAP consists of finding the 
assignment, i.e.\ a bijection $\phi: A\mapsto B$ between two equally-sized sets $A$ and $B$ ($\lvert A\rvert = \lvert B\rvert)$ that optimizes the objective function of the form: 
\begin{equation}
\sum_{a \in A}c(a, \phi(a))
\label{def:LAP_c}
\end{equation}
A minimization LAP can be easily turned into the maximization version: it suffices to either change the sign of the objective function or subtract all the costs from the maximum cost. The latter technique is often used in practice, since some implementations of LAP algorithms assume all costs are non-negative.  We refer the reader to \cite{comb_opt,assignment_revised} for more details on LAP. 
Observe that finding the COPR directly by Definition~\ref{def:COPR} includes finding the process relabeling $\sigma$ that induces the relabeled graph $G_\sigma$ with minimal cost $W(G_\sigma)$. This is not directly an instance of LAP because the graph structure depends on $\sigma$. However, we will show that it can be reduced to LAP by first defining the relabeling gain, then proving that maximizing the relabeling gain yields the COPR and finally using this relabeling gain to formulate the problem of finding the COPR as a Linear Program corresponding to LAP.

\begin{definition} (Relabeling Gain)
Let $G=(P, E, S)$ be a communication graph, $w$ a communication cost function, $\sigma$ an arbitrary process relabeling and let $G_\sigma$ be the relabeled graph induced by $\sigma$. The relabeling gain $\delta: P\times P \mapsto \mathbb{R}$ for some $p_x, p_y \in P$ describes the gain of relabeling $p_x \xrightarrow[]{\sigma} p_y$ and is defined as:
\begin{equation}
\delta(p_x, p_y) = \sum_{p_i \in \mathcal{N}_G^L(p_x)} (\underbrace{w(p_i, p_x, S_{i, x})}_\text{before relabeling} - \underbrace{w(p_i, p_y, S_{i,x})}_\text{after relabeling})
\label{eq:delta}
\end{equation}
The total relabeling gain is defined as the sum of relabeling gains for each process:
$$\Delta_\sigma = \sum_{p_j \in P} \delta(p_j, p_{\sigma(j)})$$
\label{def:gain}
\end{definition}
\begin{remark}
If $w$ is the \emph{locally-free-volume-based} cost function defined in Equation~\eqref{eq:locally_free_w}, then it is easy to see that: $$\delta(p_x, p_y) = V(S_{y, x}) - V(S_{x,x}),$$
which intuitively means that by relabeling $p_x\xrightarrow[]{\sigma}p_y$, we gained $S_{y, x}$ as it became a local exchange (which costs 0 under $w$) but we lost $S_{x,x}$ which after relabeling requires a remote communication.
\label{rmk:gain}
\end{remark}

In the following Lemma we prove that the total relabeling gain is equal to the total weight difference between $G$ and $G_\sigma$, i.e.\ before and after the process relabeling.

\begin{lemma}
Let $G=(P, E, S)$ a communication graph, $w$ a communication cost function and $\sigma$ an arbitrary process relabeling. If $G_{\sigma}=(P, E', S')$ is the relabeled graph induced by $\sigma$ and $\Delta_\sigma$ the total relabeling gain, then the following holds:
$$\Delta_\sigma = W(G) - W(G_\sigma)$$
\label{lemma:delta_s}
\end{lemma}
\begin{proof}
By Equation~\eqref{def:WG} and Definition~\ref{def:RG}, for $W(G_\sigma)$ we have:
\begin{align}
W(G_\sigma) &= \sum_{(p_i, p_j) \in E'} w(p_i, p_j, S'_{ij})=\sum_{(p_i, p_j) \in E} w(p_i, p_{\sigma(j)}, S_{ij})\nonumber \\
&=\sum_{\vphantom{0^L}p_j \in P}\ \sum_{p_i \in \mathcal{N}^L_G(p_j)} w(p_i, p_{\sigma(j)}, S_{ij}) \label{lemma:eq:WGs}
\end{align}
Similarly, for $W(G)$, we have:
\begin{equation}
W(G)=\sum_{\vphantom{0^L}p_j \in P}\ \sum_{p_i \in \mathcal{N}^L_G(p_j)} w(p_i, p_j, S_{ij})
\label{lemma:eq:WG}
\end{equation}
Subtracting equations \eqref{lemma:eq:WG} and \eqref{lemma:eq:WGs}, by Definition~\ref{def:gain}, we get:
$$W(G) - W(G_\sigma) = \sum_{\vphantom{0^L}p_j \in P}\ \underbrace{\sum_{p_i \in \mathcal{N}^L_G(p_j)}  (w(p_i, p_j, S_{ij})- w(p_i, p_{\sigma(j)}, S_{ij}))}_{=\delta(p_j, \sigma(p_j))}= \Delta_{\sigma}.$$
\end{proof}

Next, we show that the COPR can also be obtained by maximizing the total relabeling gain. 

\begin{lemma}
Let $G=(P, E, S)$ be a communication graph and $w:P\times P\times S\mapsto \mathbb{R}$ a communication cost function. A communication-optimal process relabeling $\sigma_{opt}$ with respect to the communication function $w$ is given by:
$$\sigma_{opt} = \argmax_\sigma \Delta_\sigma$$
\label{lemma:argmax}
\end{lemma}
\begin{proof}
By Lemma~\ref{lemma:delta_s}, we have:
$$\argmax_\sigma \Delta_\sigma = \argmax_\sigma (W(G) - W(G_\sigma)).$$
Observe that $W(G)$ does not depend on $\sigma$ and is therefore constant with respect to $\sigma$. Hence, we can write:
$$\argmax_\sigma (W(G) - W(G_\sigma)) = \argmax_\sigma (-W(G_\sigma)) = \argmin_\sigma W(G_\sigma)  =\sigma_{opt},$$
where the last equality follows from Definition~\ref{def:COPR}.
\end{proof}
Finally, we prove that finding the COPR can be reduced to the following Linear Program (LP) that corresponds to the Linear Assignment Problem (LAP). 
\begin{theorem}
Let $P=\{p_0, \dots, p_{n-1}\}$ be a set of processes, $G=(P, E, S)$ a communication graph, $w: P \times P \times S\mapsto \mathbb{R}$ a communication cost function and $\delta: P\times P \mapsto \mathbb{R}$ a relabeling gain. Let $x^*_{ij}$ ($i, j=0, 1,\dots n$) be the optimal solution to the following Linear Program:
\begin{alignat}{2}
    \text{maximize } &\sum_{(p_i, p_j) \in P\times P} \delta(p_i, p_j) x_{ij} \label{LR:obj} \\
    \text{subject to: } \nonumber \\
    &\sum_{p_i \in P} x_{ij} = 1 &j=0, \dots, n-1 \label{eq:LP1} \\
    &\sum_{p_j \in P}x_{ij} = 1 &i = 0, \dots, n-1 \label{eq:LP2} \\
    & x_{ij} \geq 0 &i,j=0, \dots, n-1 \label{eq:LP3}.
\end{alignat}
The communication-optimal process relabeling (COPR) $\sigma_{opt}$ is given by:
$$\sigma_{opt}(i) = j \Leftrightarrow x^*_{ij} = 1 \text{ for all } i, j=0,\dots,n-1$$
\label{thm:assignment}
\end{theorem}
\begin{proof}
This LP corresponds to the Linear Assignment Problem \cite{assignment_revised}. Due to Birkhoff \cite{birkhoff} (also reformulated as Theorem 1.1 in \cite{comb_opt}), we can assume $x_{ij} \in \{0, 1\}$, as this condition can be relaxed to $x_{ij} \geq 0$ in this case. The conditions \eqref{eq:LP1} and \eqref{eq:LP2} ensure the induced process relabeling is a bijection. Therefore, each feasible solution of this LP is a matrix representation of the permutation it induces.

The stated LP is always feasible, because $x_{ij} = 1 \Leftrightarrow i=j$ for all $i,j \in [n]$, corresponding to the identity permutation, is always a feasible solution. Let $x'_{ij}$ be an arbitrary feasible solution of the LP and let $\sigma$ be the permutation induced by $x'_{ij}$. Since $x'_{ij} \in \{0, 1\}$ and $x'_{ij} = 1 \Leftrightarrow \sigma(i)=j$, the objective function becomes:
\begin{equation}
\sum_{(p_i, p_j) \in P\times P} \delta(p_i, p_j) x'_{ij} = \sum_{p_i \in P} \delta(p_i, p_{\sigma(i)}) =\Delta_\sigma \label{eq:obj1}
\end{equation}

By Lemma~\ref{lemma:argmax}, maximizing the total relabeling gain $\Delta_\sigma$ yields the communication optimal process relabeling $\sigma_{opt}$ which finalizes the proof.
\end{proof}

Since the Linear Assignment Problem can also be formulated in terms of Graph Matchings \cite{comb_opt}, we also provide the equivalent reformulation of Theorem~\ref{thm:assignment} in terms of the \emph{Maximum Weight Bipartite Perfect Matching} problem. 

\begin{theorem}
Let $P=\{p_0, \dots, p_{n-1}\}$ be a set of processes, $G=(P, E, S)$ a communication graph, $w: P \times P \times S\mapsto \mathbb{R}$ a communication cost function and $\delta: P\times P \mapsto \mathbb{R}$ a relabeling gain. Let $G_\delta = (P, P, E_\delta)$ be a complete bipartite graph with edges $E_\delta = P\times P$, where each edge $(p_i, p_j) \in E_\delta$ is assigned weight $\delta(p_i, p_j)$. If $M\in E_\delta$ is a Maximum Weight Perfect Matching of graph $G_\delta$, the communication-optimal process relabeling (COPR) $\sigma_{opt}$ of $G$ is given by:
\begin{equation}
\sigma_{opt}(i) = j \Leftrightarrow (p_i, p_j) \in M
\label{eq:matching}
\end{equation}
\label{thm::assignment2}
\end{theorem}
\begin{proof}
This is just a reformulation of Theorem~\ref{thm:assignment} where each feasible solution $x'_{ij}, i,j \in [n]$ of the LP from Theorem~\ref{thm:assignment} can also be viewed as a matching of $G_\delta$. The relation between $x'_{ij} (i, j\in [n])$, the corresponding relabeling $\sigma$ and the corresponding matching $M$ is given by:
$$x'_{ij} = 1 \Leftrightarrow \sigma(i) = j\Leftrightarrow (p_i, p_j) \in M.$$
\end{proof}
\begin{remark}
The graph $G_\delta$ admits a Perfect Matching because $M = \{(p_i, p_i): p_i \in P\}$ is always a valid Perfect Matching. 
\label{rmk:perfect}
\end{remark}

\subsection{COPR Algorithm}

In Theorem~\ref{thm:assignment} it is shown how finding the COPR can be reduced to the Linear Program that corresponds to the Linear Assignment Problem (LAP). A reformulation of this LAP in terms of Maximum Weight Bipartite Perfect Matching (MWBPM) yields Theorem~\ref{thm::assignment2}. In addition, the LAP can also be formulated in terms of Network Flows, in which case it is reduced to the \emph{Maximum Flow of Optimal Cost} problem \cite{comb_opt}.

An example of the matching-based algorithm for finding the COPR that follows from Theorem \ref{thm::assignment2} is shown in Algorithm~\ref{alg:copr}. The complexity of this algorithm depends on the complexity of 1) computing the weights, i.e.\ costs 2) solving a LAP (Line~\ref{alg:line:lap}). Let $\lvert P\rvert = n$ be the number of processes. 

The weights are computed in Lines~\ref{alg:line:delta}--\ref{alg:line:delta2}. If all data volumes $V(S_{ij}), i, j \in [n]$ are precomputed, then the for-loop computing all $\delta(p_i, p_j)$ by Equation~\eqref{eq:delta} serially takes $O(n^3)$. On distributed architectures, this reduces to $O(n^2)$. Furthermore, for simpler cost functions like locally-free-volume-based cost from Equation~\eqref{eq:locally_free_w}, computing $\delta(p_i, p_j)$ is constant (see Remark~\ref{rmk:gain}), which further reduces the total complexity down to $O(n)$. The complexity is therefore dominated by the complexity of solving the LAP. 

The LAP solver is invoked in Line \ref{alg:line:lap}. One of the most famous LAP algorithms is the Hungarian (Kuhn–Munkres) Algorithm \cite{hungarian,hungarian3} with complexity $O(n^3)$, which is optimal for dense graphs that we are dealing with (note that the graph $G_\delta$ from Theorem~\ref{thm::assignment2} is a \emph{complete} bipartite graph). This algorithm has also been GPU-accelerated \cite{gpu2} and there is also a distributed version with a multi-gpu support \cite{gpu1}. Other interesting LAP algorithms include a fast matching-based randomized algorithm \cite{angelika} and a recently developed, distributed, approximation algorithm \cite{approx} that achieves great speedups while finding near-optimal solutions.

\begin{algorithm}
\caption{Finding the COPR}
\begin{algorithmic}[1]
\Require 
\Statex Process Set: $P=\{p_0, p_1, \dots, p_{n-1}\}$ 
\Statex Data Set: $S=[[S_{ij}]]_{0 \leq i,j < n}$ \Comment{$S_{ij}:=\text{package to be sent }p_i\mapsto p_j$}
\Statex Communication-cost function: $w$ \Comment{$w(p_i, p_j, S_{ij}) := \text{cost of sending } p_i \xrightarrow[]{S_{ij}} p_j$}
\Ensure 
\Statex Comm-Optimal Process Relabeling (COPR): $\sigma_{opt}: [n]\mapsto [n]$ \Comment{$p_i\rightarrow p_{\sigma_{opt}(i)}$}
\Procedure{FindCOPR}{$P, S, w$} $\rightarrow \sigma_{opt}$
    \State $\sigma_{opt} = \textbf{0}_n$ \Comment{COPR as an array of size $n$}
    \State $weights=\textbf{0}_{n \times n}$ \label{alg:line:delta} 
    \Comment{adjacency matrix of $G_{\delta}$ from Theorem~\ref{thm::assignment2}}
    \For{$(p_i, p_j) \in P\times P$} \Comment{$G_\delta$ is a complete bipartite graph} 
        \State $weights[i][j] = \delta(p_i, p_j)$ \Comment{$\delta(\cdot, \cdot)$ defined in Equation~\eqref{eq:delta}}
    \EndFor \label{alg:line:delta2}
    \State $M = \text{MWBPM}(n, weights)$ \Comment{Max Weight Bipartite Perfect Matching$(G_\delta)$} \label{alg:line:lap}
    \For{$(i, j) \in M$}
    \State $\sigma_{opt}[i] = j$ \Comment{$M\mapsto \sigma_{opt}$ as in Equation~\eqref{eq:matching}}
    \EndFor
    \State \Return $\sigma_{opt}$
\EndProcedure
\end{algorithmic}
\label{alg:copr}
\end{algorithm}

\section{COSTA: Comm-Optimal Shuffle and Transpose Alg.}
\label{sec:costa}

COSTA uses the communication-optimal process relabeling (COPR) to implement the routine:
\begin{equation}
A=\alpha\cdot \operatorname{op}(B) + \beta \cdot A, \quad \operatorname{op} \in \{\operatorname{transpose}, \operatorname{conjugate-transpose}, \operatorname{identity}\}
\label{eq:costa}
\end{equation}
on distributed systems, where $A, B$ are matrices with potentially different layouts and $\alpha, \beta$ are scalars. Since this routine, in a distributed setting, includes the data reshuffling (i.e.\ redistribution) across processes while potentially transposing the data, we call this routine: \emph{Shuffle and Transpose}. It encapsulates the functionality of two ScaLAPACK routines: \texttt{pxtran(u)} for matrix transpose and \texttt{pxgemr2d} for data redistribution.

\phantomsection
\label{par:matrix_layout}
\textbf{Matrix Layout:} describes how the matrix is distributed among processes. The way how a matrix $A$ is partitioned is given by two sorted arrays: row-splits $R_A$ and column-splits $C_A$ where block $b_{ij}$ contains the rows in range $\left[R_A(i), R_A(i+1)\right)$ and the columns in range $\left[C_A(j), C_A(j+1)\right)$. The row-splits and columns-splits together define the grid $\mathrm{Grid}_A = (R_A, C_A)$. The owner of a block $b$ from this grid is given by $\operatorname{Owners}_A(b)\in P$, where $P$ is the set of processes holding $A$. The layout of a matrix $A$ is hence an ordered tuple $L(A)=(\mathrm{Grid}_A, P, \operatorname{Owners}_A)$. 

\textbf{Grid Overlay:} given two grids $\mathrm{Grid}_A=(R_A, C_A)$ and $\mathrm{Grid}_B=(R_B, C_B)$ we define the Grid Overlay $\mathrm{Grid}_{A,B}=(R_A \cup R_B, C_A \cup C_B)$ as the grid obtained by overlaying both grids. It is easy to see that each block $b_{A, B} \in \mathrm{Grid}_{A, B}$ is covered by one and only one block $b_A\in \mathrm{Grid}_A$ and one and only one block $b_B\in \mathrm{Grid}_B$. We therefore define $\operatorname{cover}_{A}(b_{A, B}) = b_A$ and $\operatorname{cover}_{B}(b_{A, B}) = b_B$. 

\textbf{Data Reshuffling:} given matrices $A$ and $B$ with same dimensions, but different layouts $L(A)$ and $L(B)$ on processes $P=\{p_0, p_1,\dots,p_{\lvert P\rvert}\}$, we want to copy the values of $B$ into the layout of $A$. In a distributed setting, this includes communication of matrix pieces. In order to be able to use Algorithm~\ref{alg:copr} for obtaining the communication-optimal process relabeling (COPR) $\sigma_{opt}$ for this problem, we have to construct the set of packages $S=[[S_{ij}]_{0\leq i,j<\lvert P\rvert}$ where $S_{ij}$ contains all blocks that should be sent from process $p_i$ to $p_j$. We show how to obtain the COPR for this problem in Algorithm~\ref{alg:costa_copr}.

\begin{algorithm}
\caption{COPR for (Matrix) Data Reshuffling}
\begin{algorithmic}[1]
\Require 
\Statex Matrix Layout $L(A)=(\mathrm{Grid}_A, P, \operatorname{Owners}_A)$
\Statex Matrix Layout $L(B)=(\mathrm{Grid}_B, P, \operatorname{Owners}_B)$
\Statex Communication-Cost Function $w$ 
\Ensure
\Statex Set of data packages: $S=[[S_{ij}]]_{0\leq i,j <\lvert P\rvert}$ \Comment{package $S_{ij}$ to be sent from $p_i$ to $p_j$}
\Statex COPR: $\sigma_{opt}$ for copying $B$ into the layout of A \Comment{$\sigma_{opt}$ relabeling: $p_i\rightarrow p_{\sigma_{opt}(i)}$}
\Procedure{FindCOPRforMatrices}{$L(A), L(B), w$} $\rightarrow (S,\sigma_{opt})$
    \State $S=[[S_{ij}]]_{0\leq i,j <\lvert P\rvert}=\emptyset_{\lvert P\rvert \times \lvert P \rvert}$ \Comment{set of data packages: initialize all $S_{ij}=\emptyset$}
    \For{$b \in \mathrm{Grid}_{A,B}$} \Comment{iterate over the Grid Overlay}
        \State $p_i=\operatorname{Owners}_A(\operatorname{cover}_A(b))$ \Comment{owner of the block which covers $b$ in $L(A)$}
        \State $p_j=\operatorname{Owners}_B(\operatorname{cover}_B(b))$ \Comment{owner of the block which covers $b$ in $L(B)$}
        \State $S_{ij}=S_{ij} \cup \{b\}$ \Comment{add block $b$ to the right package}
    \EndFor 
    \State $\sigma_{opt}=FindCOPR(P, S, w)$ \Comment{Algorithm~\ref{alg:copr}}
    \State \Return (S, $\sigma_{opt})$
\EndProcedure
\end{algorithmic}
\label{alg:costa_copr}
\end{algorithm}

\textbf{Scale and Transpose/Conjugate:} observe that the routine from Equation~\eqref{eq:costa} includes the possibility to scale the matrices (multiply by a scalar), transpose/conjugate or take a submatrix. Let $L(A)$ and $L(B)$ be two different matrix layouts. In the previous paragraph we discussed the case when matrix $B$ should be copied to the layout of matrix $A$ without any transformation. Here, we discuss the cases when $B$ should also be transformed. If only a submatrix of $B$ should be taken, then we can first truncate the corresponding row-splits and column-splits in $\mathrm{Grid}_B$ and then apply the Algorithm~\ref{alg:costa_copr} to obtain the COPR. If $B$ should also be transposed/conjugated or scaled before being copied to $A$, then the cost of this transformation can be taken into account in the communication-cost function $w$ (see Section~\ref{sec:comm_cost}). Practically, the transformation can be performed in one of the following ways:
\begin{itemize}
    \item transform before sending: each process can first transform the data locally in temporary send buffers, and then send it.
    \item transform after receiving: each process can transform the data upon receipt in temporary receive buffers. This approach is better in asynchronous settings because the data transformation can be overlapped with communication of other packages.
\end{itemize}
We chose to transform upon receipt, since we are using asynchronous communication.

\textbf{Communication-Cost Function:} the communication cost function can be arbitrary. In practice, we use the simple locally-free-volume-based cost function defined in Equation~\eqref{eq:locally_free_w}. 


\noindent Finally, taking into account these insights, we present COSTA in Algorithm~\ref{alg:costa}.

\begin{algorithm}
\caption{COSTA: Comm-Optimal Shuffle and Transpose Algorithm}
\begin{algorithmic}[1]
\Require 
\Statex Matrix $A$ with Layout $L(A)=(\mathrm{Grid}_A, P, \operatorname{Owners}_A)$
\Statex Matrix $B$ with Layout $L(B)=(\mathrm{Grid}_B, P, \operatorname{Owners}_B)$
\Statex Scalars $\alpha, \beta$
\Statex Operator $\operatorname{op} \in \{\operatorname{transpose}, \operatorname{conjugate-transpose}, \operatorname{identity}\}$
\Statex Comm-Cost Function (optional) $w$ \Comment{by default, defined by Equation~\eqref{eq:locally_free_w}}
\Statex Process id: $p_{id} \in P$
\Ensure
\Statex Performs $A = \alpha\cdot \operatorname{op}(B) + \beta \cdot A$ 
\Procedure{COSTA}{$A, L(A), B, L(B), \alpha, \beta, op$}
    \State $(S, \sigma_{opt})=\operatorname{FindCOPRforMatrices}(L(A), L(B), w)$ \label{line:s} \Comment{Algorithm~\ref{alg:costa_copr}}
    \For{$p_j \in P$} 
        \State send asynchronously $S_{id, j}$ to $p_{\sigma(j)}$ \Comment{send local data to relabeled processes}
    \EndFor
    \For{$\mathrm{package} \in \{0, 1, \dots, \lvert P \rvert\}$} 
        \State receive from any $p_i\in P$ in a temp. recv buffer
        \State scale, transpose or conjugate the received package $S_{i, id}$ if needed
    \EndFor
\EndProcedure
\end{algorithmic}
\label{alg:costa}
\end{algorithm}

\section{Implementation Details}
\label{sec:implementation}

COSTA (Algorithm~\ref{alg:costa}) is implemented using the hybrid MPI+OpenMP parallelization model. The code is publicly available under the BSD-3 Clause Licence at~\cite{costa}. It has the following features: 1) provides the ScaLAPACK wrappers for \texttt{pxgemr2d} and \texttt{pxtran}; 2) supports arbitrary grid-like matrix layouts (not limited to block-cyclic). It also support both row- and col-major ordering of matrices, unlike ScaLAPACK; 3) can use the COPR to minimize the communication; 4) supports batched transformation, i.e., multiple pairs of matrix layouts can be transformed in the same communication round; 5) supports arbitrary data types using C++ templates.

\begin{figure}
\centering
\includegraphics[width=0.7\textwidth]{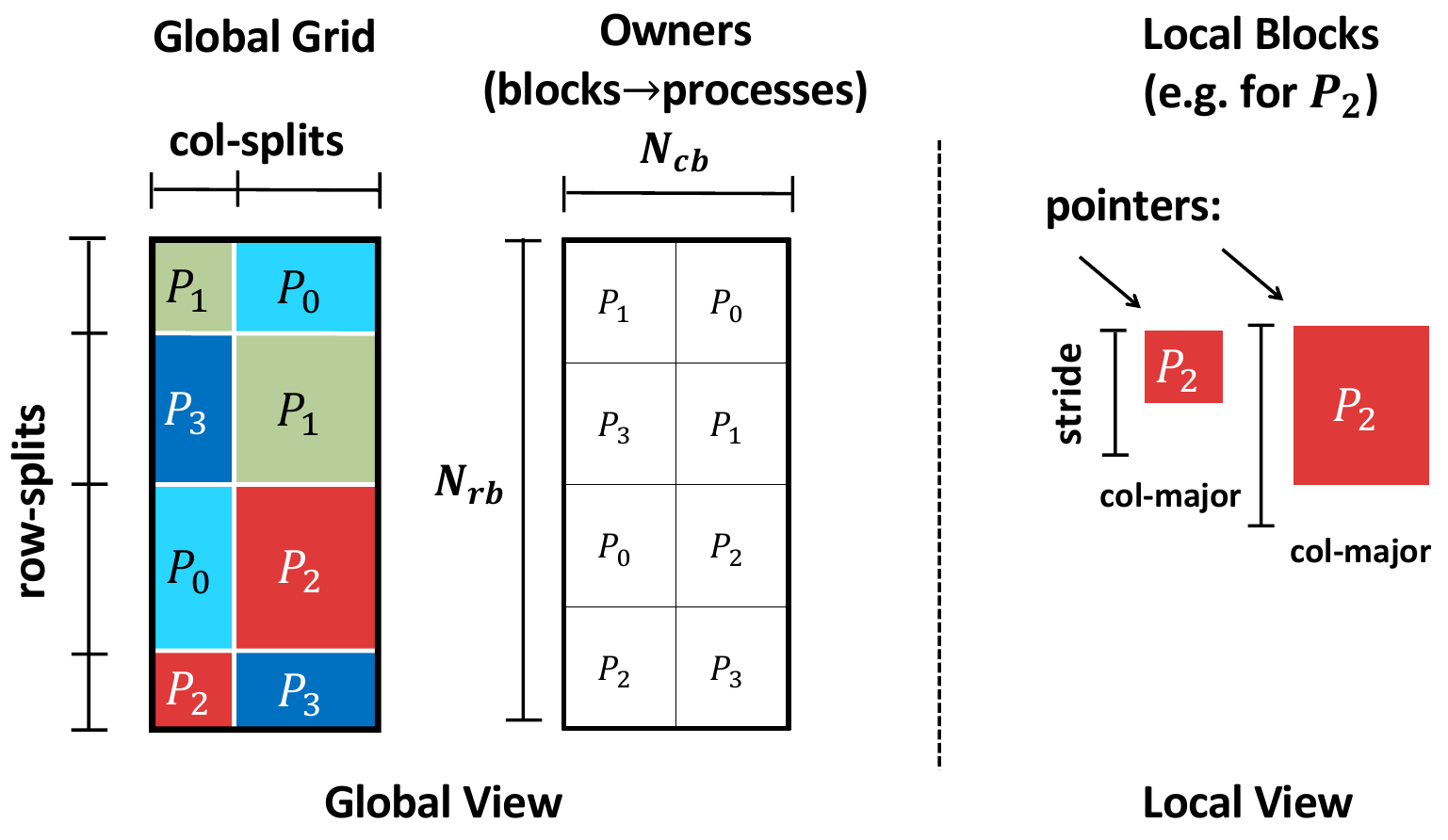}
\caption{COSTA Matrix Descriptor: a global view contains a grid (defined by row-splits and col-splits) and the owners matrix mapping blocks to processes. A local view is a list of blocks, each containing a pointer, a stride, dimensions and the ordering type (row- or col-major).}
\label{fig:matrix}
\end{figure}

\textbf{Matrix Layout Descriptor.} Following the theoretical definition of a matrix layout from Paragraph~\ref{par:matrix_layout}, in practice we use a more detailed matrix descriptor depicted in Figure~\ref{fig:matrix}, that also takes into account block strides and whether the blocks are stored in row- or col-major order. This makes the COSTA layout descriptor more general than the block-cyclic descriptor used by ScaLAPACK.

\textbf{Implementation.} After $S=[[S_{ij}]]$ is computed in Line~\ref{line:s}, each process has a list of blocks it should send and a list of blocks it should receive. Each process first copies all blocks into a temporary send buffer using OpenMP, such that all blocks to be sent to the same target are packed together into a single, contiguous package in the send buffer. These packages are then sent using the non-blocking $\texttt{MPI\_Isend}$ such that only a single package (containing multiple blocks) is sent to each receiving process, which significantly reduces the latency costs.

\textbf{Overlap of Communication and Computation.} Upon receipt of a package with $\texttt{MPI\_Waitany}$, blocks within this package are unpacked and the transformation (transpose, conjugate-transpose or multiplication by a scalar) is performed using OpenMP, while other packages are still being communicated in the background, thus enabling the overlap of communication and computation. A cache-friendly, multi-threaded kernel for matrix transposition is provided. Moreover, the blocks that are local in both the initial and the final layouts are not copied to temporary send/receive buffers, but are handled separately with OpenMP, to avoid unnecessary data copies to and from temporary buffers and potentially additional copies that MPI might perform. The handling of these blocks is also overlapped with MPI communication. 

\textbf{Batched Transformation.} If multiple pairs of layouts are to be transformed together, then the procedure is the same, except that now a package contains blocks that might belong to different layouts. Still, all the blocks to be sent to the same process are packed together and sent within a single message, thus reducing the latency costs even further.

\textbf{Max Weight Bipartite Perfect Matching.} In Line \ref{alg:line:lap} in Algorithm \ref{alg:copr} we are free to choose how we want to solve the matching problem. In practice, we use a simple greedy algorithm, which is a 2-approximation.


\section{Performance Results}
\label{sec:performance}

We evaluate the performance of COSTA in three groups of benchmarks: first, since COSTA implements \texttt{pxgemr2d} (data-reshuffling, i.e.\ distributed copy) and \texttt{pxtran} (transpose) routines, we compare the performance of COSTA vs. available ScaLAPACK implementations (Intel MKL, Cray LibSci) in isolation (Section~\ref{sec:bench1}). This comparison is done without using the Process Relabeling, as ScaLAPACK API does not support it. In the second part (Section~\ref{sec:bench2}), we measure how much communication-volume can be reduced by using the process relabeling. In the final benchmark (Section~\ref{sec:bench3}), we run COSTA within a real-world application and analyse its performance, as well as the communication-volume reduction.

\textbf{Hardware Details:} all the experiments are performed on Piz Daint Supercomputer of Swiss National Supercomputing Centre (CSCS) which consists of two partitions: the CPU partition with 1813 nodes (Cray XC40 compute nodes, Intel Xeon E5, 2x18 cores, 64-128GB RAM) and the GPU partition with 5704 nodes (Cray XC50 compute nodes, Intel Xeon E5, 12 cores, 64GB RAM + NVIDIA P100 16GB GPU). 

\textbf{Software Details:} in the benchmarks, we used Cray-MPICH v7.7, Intel MKL v19.1, Cray LibSci and Cray LibSci-Acc v20.06, CP2K v7.1, COSMA v2.3 and COSTA v1.0. All the libraries were compiled with a GCC v10.1 compiler available on the Piz Daint Supercomputer. 

\subsection{COSTA vs. ScaLAPACK}
\label{sec:bench1}

We compare the performance of COSTA, Intel MKL and Cray LibSci for the following routines: \texttt{pdgemr2d} (distributed copy, i.e.\ reshuffling) and \texttt{pdtran} (transpose). To this end, we use the ScaLAPACK wrappers that COSTA provides (see Section~\ref{sec:implementation}).

It is known that ScaLAPACK performance often varies drastically for different block sizes which determine the matrix distributions (i.e. layouts) and the optimal block size depends on the target machine~\cite{blocks}. Scientific applications usually have a default block size (e.g. in~\cite{sirius}, it is $32\times 32$) and reaching the optimal block size (which is $128\times 128$, for our applications) requires data reshuffling and potentially a transpose operation. Inspired by this example, we run the following benchmark: we vary the matrix size from 100-200k (square-case) and for each size we transform the matrix from $32\times 32$ to $128\times 128$ block size, using the pdgemr2d routine and the same for the pdtran (transpose) routine. 

We also include the batched version of COSTA for comparison. The batched version amortizes the latency costs since multiple layouts are transformed within the same communication round. This is often useful for operations like matrix multiplication which involves 3 matrices and each of them might potentially need to be transformed, as is the case in the COSMA algorithm~\cite{cosma}. To account for this scenario, we also ran the batched version of COSTA on each test-case from this benchmark with one difference: instead of transforming a single instance of each test case, we let the batched version transform 3 identical instances of each test-case and report the amortized cost per test-case instance.

This benchmark is run on 128 dual-socket CPU nodes (2x18 cores) of Piz Daint Supercomputer using 2 MPI ranks per node, 18 threads per rank and $16\times 16$ process grid. Each experiment was repeated 5 times and the best time is reported in Figure~\ref{fig:scalapack}. We observed similar performance with other matrix shapes (including rectangular matrices) and other process grids, so here we present just the square-matrices case.

\begin{figure}
\centering
\begin{subfigure}{.49\textwidth}
\includegraphics[width=\textwidth]{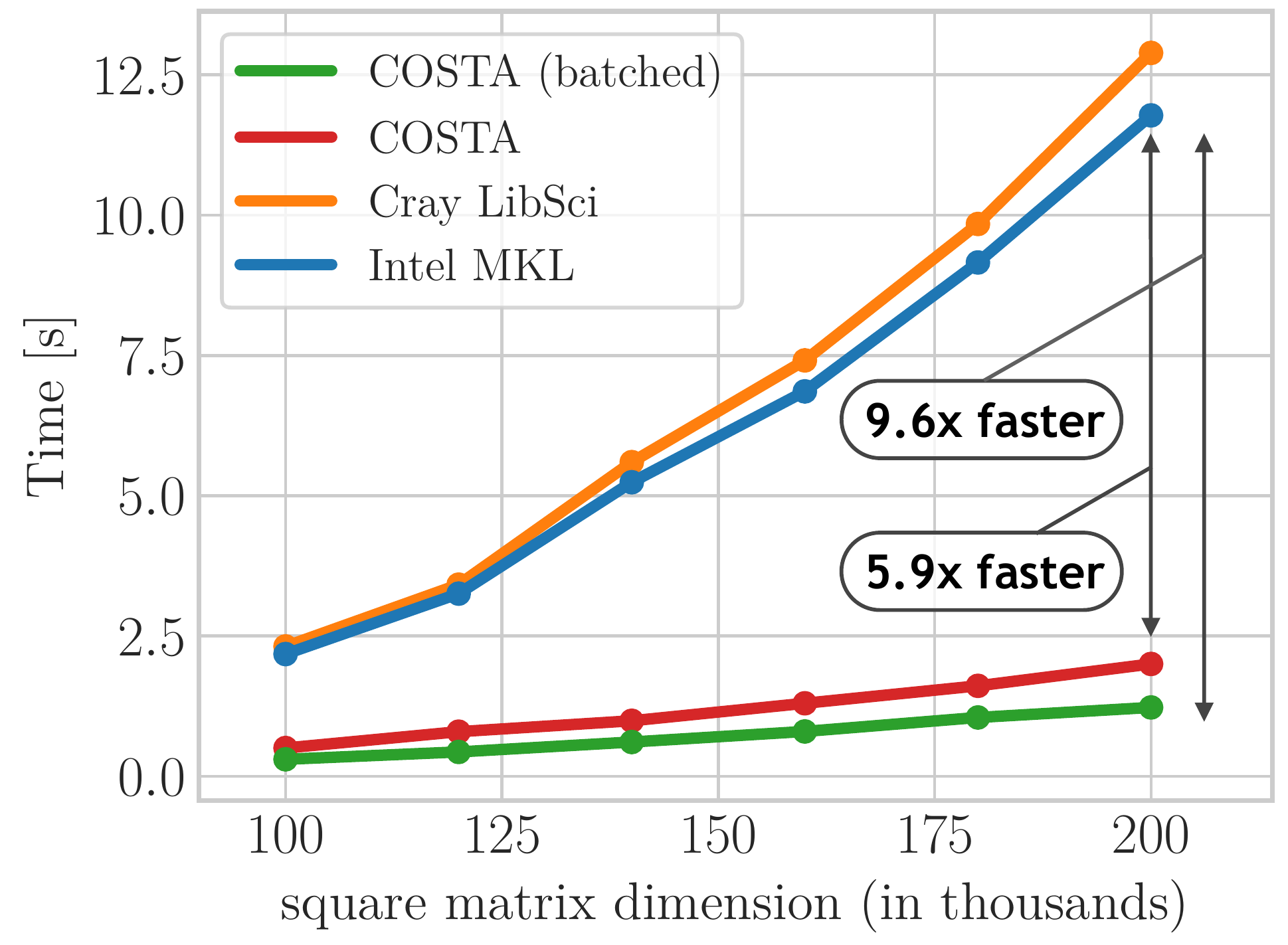}
\caption{Data Reshuffling (\texttt{pdgemr2d}).}
\end{subfigure}\hfill
\begin{subfigure}{.47\textwidth}
  \includegraphics[width=\textwidth]{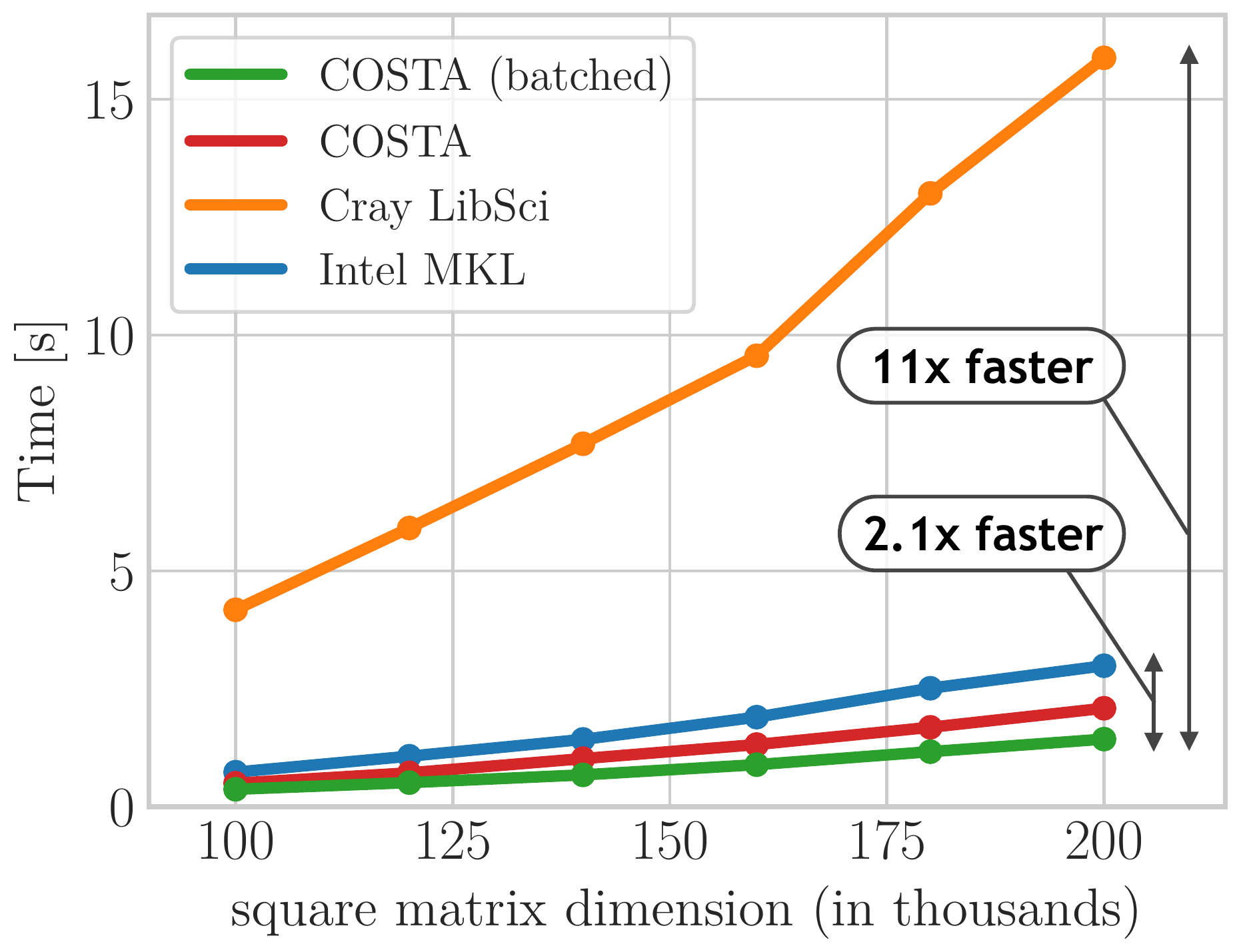}
  \caption{Matrix Transpose (\texttt{pdtran}).}
\end{subfigure}
  \caption{The performance comparison of ScaLAPACK routines for data reshuffling (left) and matrix transpose (right) by different algorithms (see Section~\ref{sec:bench1}). Reported times for COSTA (batched) are amortized over 3 identical instances of a test-case.}
\label{fig:scalapack}
\end{figure}

\subsection{Process Relabeling}
\label{sec:bench2}

We measured how much communication-volume is reduced when redistributing (reshuffling) a matrix between two block-cyclic layouts. The matrix size was $10^5\times 10^5$ and the process grid was $10\times 10$. The process grid was row-major for the initial layout and column-major for the final layout. The target layout had block size fixed at $10^4$. The block size of the initial layout was varied from 1 up to $10^4$ and for each block size the communication-volume was computed before and after relabeling. Based on this we computed the communication-volume reduction (in percent) due to process relabeling, that is shown in Figure~\ref{fig:reduced_comm}. When both layouts have the same block size ($=10^4$), then they only differ in the block assignment to processes, in which case the process relabeling is able to completely eliminate the communication.

\begin{SCfigure}
  \centering
  \includegraphics[width=0.6\textwidth]{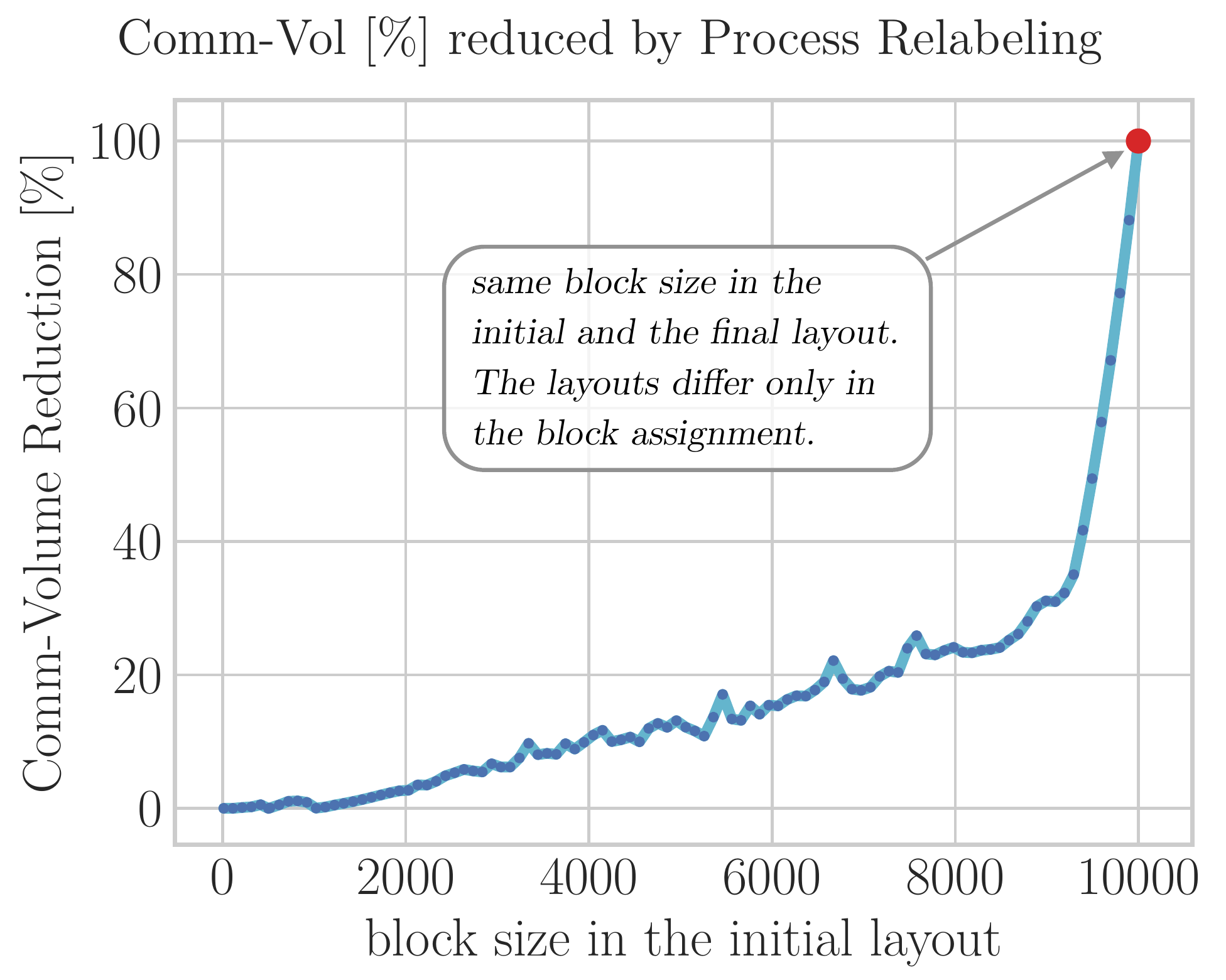}
  \vspace{0.5em}
  \caption{\label{fig:reduced_comm}
  The initial layout of the matrix (i.e.\ the block size) was varied, whereas the final layout was fixed with a block size of $10^4$. The process grid was fixed to $10\times 10$ in both layouts. When the initial and target layouts have the same blocks (and hence the same grids), process relabeling makes all communication local, thus eliminating the need for remote communication (the red dot).}
\end{SCfigure}

\subsection{Real-World Application: RPA simulations}
\label{sec:bench3}

In Random Phase Approximation (RPA) \cite{rpa} simulations, a major part of computation consists of many large tall-and-skinny matrix multiplications with matrix transposition. For a system size $N$ (the number of atoms), the matrices to be multiplied are of size $O(N^2)\times O(N)$, where $O(N^2)$ is proportional to $(\text{occupied orbitals} \cdot \text{virtual orbitals})$ and $O(N)$ is proportional to the number of auxiliary basis functions. Concretely, for simulating 128 water molecules, the matrix sizes are depicted in Figure~\ref{fig:rpa_matrix}. This multiplication is repeated many times and takes $\approx 80\%$ of the total simulation time on 128 dual-socket CPU nodes of Piz Daint. Therefore, an efficient matrix-multiplication routine is essential for this benchmark. Recently, a communication-optimal matrix-multiplication algorithm COSMA~\cite{cosma} has been developed which offers significant speedups for all matrix shapes. However, COSMA natively uses a specialized blocked matrix layout which depends on matrix dimensions and the number of available processors. Moreover, as shown in Figure~\ref{fig:rpa_matrix}, one of the matrices (matrix $A$) also has to be transposed during the reshuffling. On the other hand, the CP2K \cite{cp2k} software package, which implements the RPA method, assumes a block-cyclic (ScaLAPACK) layout. Since COSMA layout is not block-cyclic, existing ScaLAPACK routines for reshuffling and transpose cannnot be used.

We used COSTA with Process Relabeling to integrate COSMA into CP2K and compare its performance to Intel MKL, Cray LibSci (CPU) and Cray LibSci (GPU). We simulate 128 water molecules with the RPA method on 128, 256, 512 and 1024 GPU nodes of Piz Daint Supercomputer. The total matrix-multiplication time as well as the total simulation time are reported in Figure~\ref{fig:cp2k-timings}.

\begin{figure}
  \centering
  \includegraphics[width=\textwidth]{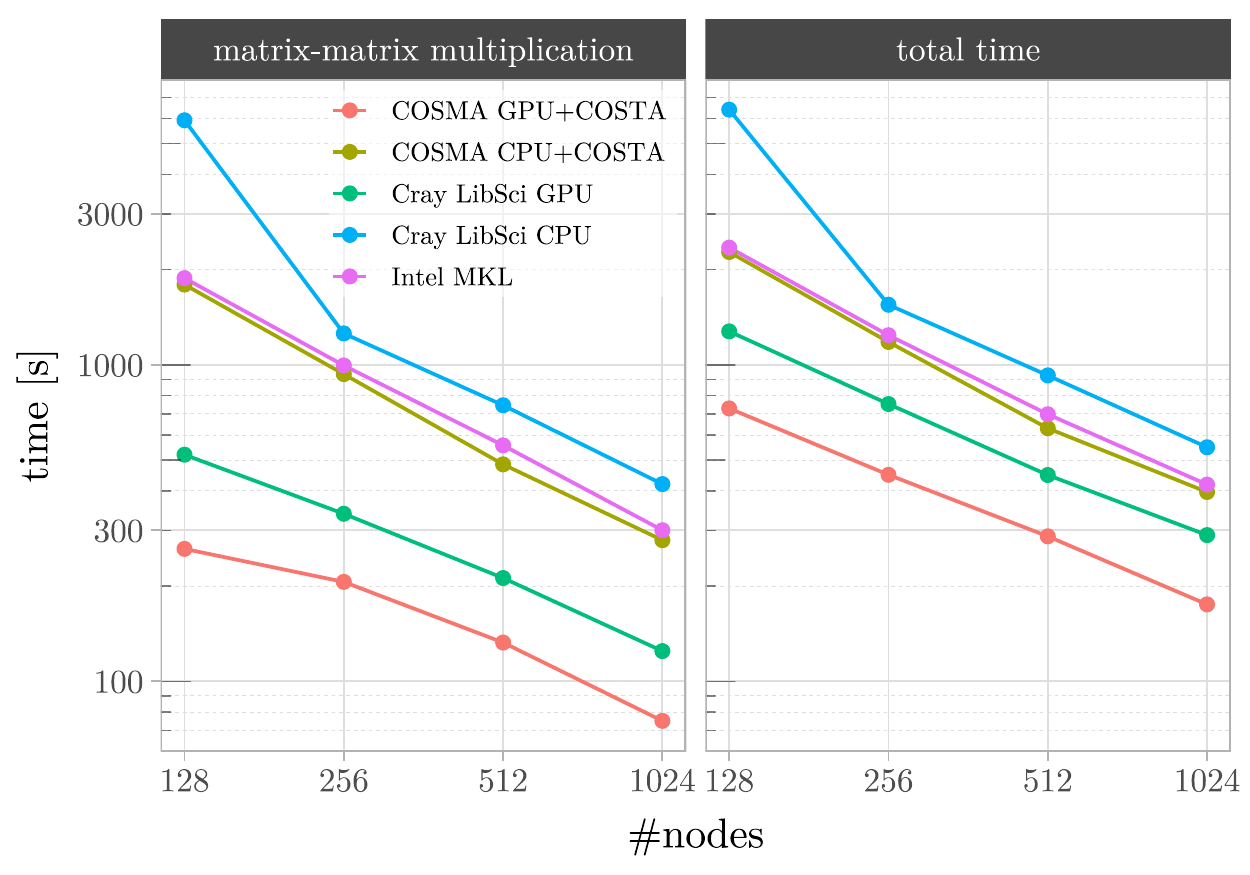}
  \caption{The RPA simulation of 128 water molecules using different matrix-multiplication backends. We used COSTA with Process Relabeling to redistribute and transpose matrices between ScaLAPACK (block-cyclic) and the native COSMA layout in each matrix-multiplication call. COSMA + COSTA outperform alternative libraries on both CPU and GPU. COSTA accounts for roughly 10\% of the total runtime of COSMA+COSTA in these cases.}
  \label{fig:cp2k-timings}
\end{figure}

\begin{figure}
  \centering
  \begin{minipage}[t]{.48\textwidth}
    \includegraphics[width=\textwidth]{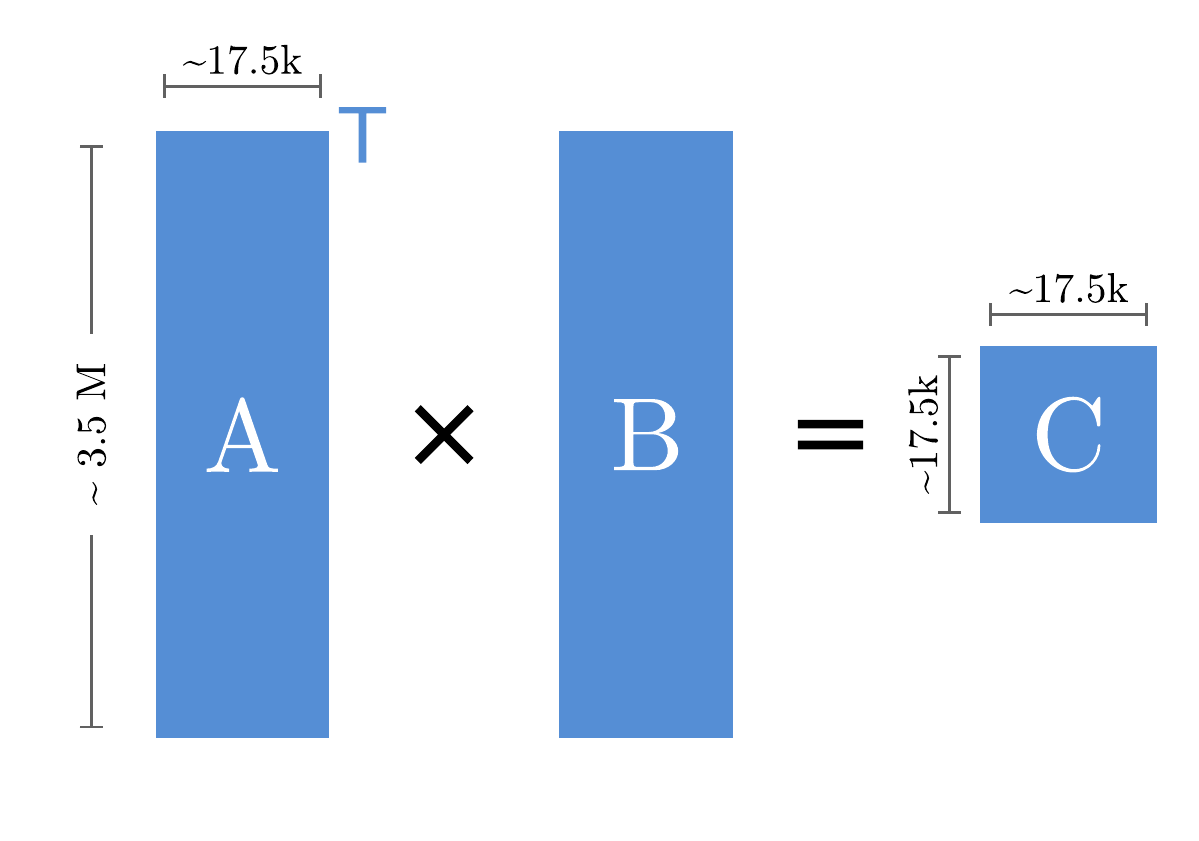}
    \caption{The computationally dominant matrix multiplication in the RPA simulation that is performed many times throughout the simulation. The exact size of $A$ and $B$ is $3,473,408 \times 17,408$.}
    \label{fig:rpa_matrix}
  \end{minipage}
  \hspace{2pt}
  \begin{minipage}[t]{.48\textwidth}
    \includegraphics[width=\textwidth]{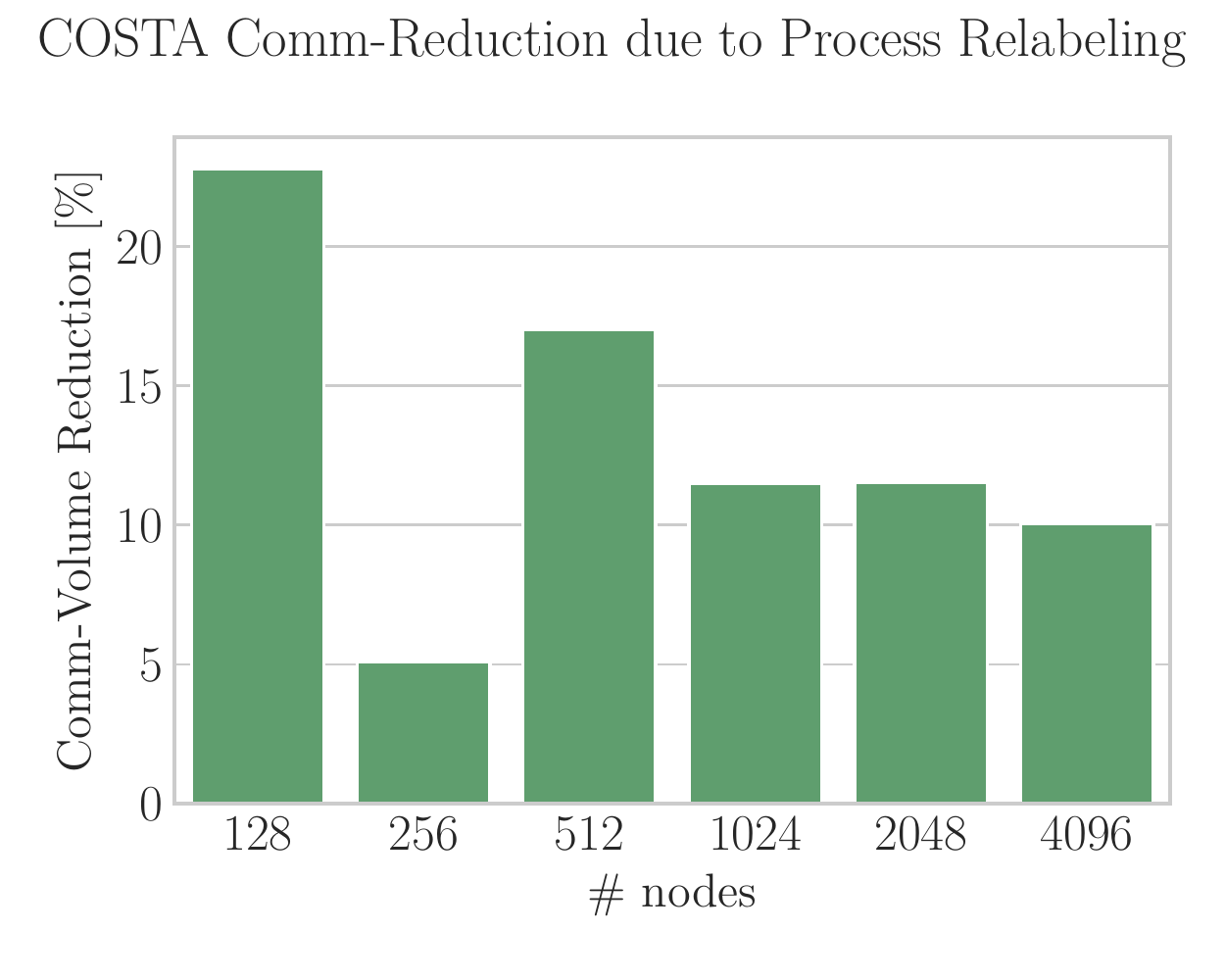}
    \caption{The communication volume reduction due to process relabeling in COSTA during the transformation of matrices (left) between the ScaLAPACK (block-cyclic) and the native COSMA layouts.}
    \label{fig:comm_volume}
  \end{minipage}
\end{figure}

For the dominant matrix multiplication in the RPA simulation (Figure~\ref{fig:rpa_matrix}), we computed the total communication volume for transforming matrices between ScaLAPACK (block-cyclic) and the native COSMA layout using COSTA. The total communication volume with and without process relabeling is measured and the percentual volume reduction is shown in Figure~\ref{fig:comm_volume}. In this case, ScaLAPACK is always using the same block sizes (and same layouts) for $A$ and $B$, whereas matrix C is distributed only on a subset of processes (the ones in the upper part of the rectangular process grid). COSMA on the other hand uses different blocks and layouts for each matrix and all matrices are distributed on all the processes. This makes it hard to predict how this interplay of different layouts is behaving as the number of nodes increases.




\section{Conclusion}

We have shown how the communication-optimal process relabeling (COPR) can efficiently be found in a very general setting where the network topology, data-locality, data transformation cost (e.g.\ transposing the data) and other parameters can all be taken into account through a cost-function. The theoretical contribution of this paper is not limited to matrix redistribution or transposition, but can also be used in general, e.g. for tensors. Besides transposition, any other operation can be performed on the local data -- it suffices to include the operation cost into the cost function of the COPR algorithm (Algorithm~\ref{alg:copr}). 

We developed COSTA: a highly-efficient algorithm with process relabeling for performing matrix shuffle and transpose routine (Equation~\eqref{eq:costa}). The experiments have shown that COSTA outperforms ScaLAPACK multiple times even when no process relabeling is used. COSTA provides ScaLAPACK wrappers for \texttt{pxgemr2d} and \texttt{pxtran} routines making the integration into scientific libraries straightforward. In addition, COSTA can also deal with arbitrary grid-like matrix layouts and is not limited to block-cyclic layouts and supports both row- and column-major storage of blocks and efficiently overlaps communication and computation. Moreover, a batched version is also provided, which can transform multiple pairs of layouts, while significantly reducing the latency.

Furthermore, we have shown that the process relabeling can reduce the communication volume for data reshuffling even by 100\%, e.g. when the initial and final layouts differ up to a process permutation. We used COSTA to integrate the highly-optimized COSMA algorithm into the CP2K software package and have shown that COSTA can enable the interoperability between different existing scientific libraries and the novel efficient algorithms with very little overhead. In practice, we have shown that COSTA is able to significantly reduce the communication cost also in real world applications where initial and final layouts are both changing in different ways.

\bibliographystyle{splncs04.bst}
\bibliography{refs.bib}

\end{document}